\documentclass[11pt,a4paper]{article}

\usepackage[T1]{fontenc}
\usepackage{lmodern}
\usepackage[margin=25mm]{geometry}
\usepackage{amsmath,amssymb,amsthm}
\usepackage{graphicx}
\usepackage{xcolor}
\usepackage{braket}
\usepackage{placeins}
\usepackage{microtype}
\usepackage[hidelinks]{hyperref}
\graphicspath{{figures/}}
\hypersetup{
  pdftitle={Genuine Information Backflow Requires Noncommuting Initial System States},
  pdfauthor={Abdullah Alsharari and Francesco Buscemi},
  pdfsubject={Quantum Physics; open quantum systems},
  pdfkeywords={non-Markovian dynamics, information backflow, quantum mutual information}
}

\newcounter{defthM}
\newenvironment{definition}[1]{%
  \refstepcounter{defthM}\trivlist
  \item[\hskip\labelsep\bfseries #1~\arabic{defthM}.]\normalfont\ignorespaces
}{\endtrivlist}
\newenvironment{theorem}[1]{%
  \refstepcounter{defthM}\trivlist
  \item[\hskip\labelsep\bfseries #1~\arabic{defthM}.]\itshape\ignorespaces
}{\endtrivlist}

\title{Genuine Information Backflow Requires
Noncommuting Initial System States}

\author{Abdullah Alsharari\thanks{\texttt{alsharari.abdullah.f1@s.mail.nagoya-u.ac.jp}}\quad and\quad Francesco Buscemi\thanks{\texttt{buscemi@nagoya-u.jp}}\\[0.5em]
{\small\itshape Graduate School of Informatics, Nagoya University, Nagoya 464-8601, Japan}}
\date{}

\begin{document}

\maketitle
  
\begin{abstract}
The Gorini--Kossakowski--Lindblad--Sudarshan (GKLS) equation, a cornerstone of open-system theory, describes Markovian semigroup dynamics in which information about the initial preparation cannot revive. Beyond this regime, information revivals are widely interpreted as backflow from a non-Markovian environment. However, a decrease followed by an increase in information retrievable from the system need not imply that information has returned from the environment [PRX Quantum 6, 020316 (2025)]. It may instead arise from noise correlations unrelated to the system's initial preparation, carried by inert degrees of freedom that never interact with the system and therefore cannot mediate backflow. Such revivals are called extrinsic; genuine backflow is reserved for revivals that cannot be explained this way. The framework developed in that work, however, assumes access to arbitrary initial system states, an assumption often incompatible with physical constraints or limited experimental control. This leaves open which properties of the available preparations are necessary for an information revival to certify genuine backflow. Here we extend the framework to allow the description of restricted families of initial states and show that genuine backflow is possible only if the family contains noncommuting states. If all available states commute, as in a classical stochastic process, every revival can be explained by a suitable inert extension and is therefore extrinsic. Noncommutativity of the available system states is thus necessary for information revivals to provide an operational signature of environmental memory.

\par\smallskip
\noindent\textbf{Keywords:} Open quantum systems;
non-Markovian dynamics; information backflow;
quantum mutual information; genuine revival.

\end{abstract}

\section{Introduction}

A classical stochastic process is Markovian when, conditioned on the present state, its future is independent of its past. This property is formulated directly in terms of the joint probability distribution of the process at multiple times. It implies, in particular, that information about the initial state cannot increase as the system evolves: once such information has been lost from the present state, it cannot later reappear. This information-theoretic consequence of Markovianity follows from the data-processing inequality and has a long history in probability, statistics, physics, and information theory~\cite{Ingarden1997InformationDynamics,norris_1997,Cover2006}.

Extending this notion of Markovianity to quantum processes is not straightforward. Quantum theory does not provide a single object that directly generalizes the joint multi-time probability distribution of a classical stochastic process, at least not without introducing additional choices concerning measurements or interventions. As a consequence, several inequivalent notions of quantum Markovianity have been proposed. At the level of reduced dynamics, a landmark achievement is the Gorini--Kossakowski--Lindblad--Sudarshan (GKLS) equation, whose characterization of the generators of completely positive trace-preserving dynamical semigroups has served as a cornerstone of open-system theory for fifty years~\cite{gorini-kossakowski-sudarshan-1976,lindblad-1976-generators}.

For a quantum dynamical semigroup, evolution between any two times is described by a completely positive trace-preserving map. Divisibility-based approaches retain this property without requiring the semigroup law, leading to the broader notion of CP-divisibility~\cite{Rivas-Huelga-Plenio_2010}. This includes evolutions generated by regular time-dependent GKLS generators with nonnegative rates at every time~\cite{chruscinski-2022-dynamical-maps}. By the data-processing inequality, CP-divisible dynamics cannot increase the mutual information between the system and an inert reference. The GKLS formalism thus provides a natural benchmark for the absence of information revivals.

Other approaches use the monotonic behavior of distinguishability or correlations as an information-theoretic criterion of quantum Markovianity~\cite{Breuer-Laine-Piilo_2009,luo-2012-non-markovianity-via-correlations,megier-prl-2021,vacchini-2022-review}. Under suitable conditions, divisibility and information-theoretic criteria can be reconciled~\cite{buscemi-datta-2016-divisibility}, but their relationship depends on the information quantity and the family of preparations considered, and they are not equivalent in general.

In this paper we adopt an information-theoretic approach. The system is initially correlated with an inert reference, and the information accessible from the system is measured by its mutual information with that reference. Suppose that this quantity decreases between two times and then increases at a later time. Such an increase, which we call an \textit{information revival}, cannot be produced by a channel acting locally on the system at the intermediate time. It therefore excludes a CP-divisible description of the reduced evolution over that interval, including a description by a GKLS semigroup.

In the standard \textit{backflow} picture, information about the initial preparation
first leaves the system and becomes available only when the active environment
is retained alongside the system. A later system--environment interaction then
transfers part of that information back to the system, producing the observed
revival. The increase of correlations between system and reference is therefore
commonly interpreted as evidence for a dynamical return of information from the
environment.

An information revival, however, does not by itself prove that this is what
happened. The revival is a property of the reduced reference--system process,
whereas its explanation depends on which inaccessible degrees of freedom are
retained. Ref.~\cite{buscemi-2025-non-causal-backflow} showed that a revival may
disappear after including auxiliary systems that never interact with the open
system. Such inert systems cannot dynamically return information. They only
reveal correlations that were already present before the evolution. In this
case, the revival is extrinsic and no genuine backflow from the active
environment is needed to explain it.

This observation changes the question that should be asked about an information
revival. The relevant issue is not simply whether information increases, or by
how much, but whether the increase can be explained without retaining degrees of
freedom that actively participate in the system--environment interaction. A
revival represents genuine information backflow only when all explanations based
solely on pre-existing inert correlations can be excluded.

A complementary line of work classifies non-Markovian reduced dynamics according
to the type of memory required for its
realization~\cite{backer2024local,backer2025entropic}. The central question is
whether a sequence of reduced channels can be realized by storing only a
classical record between successive steps, or whether a persistent quantum
memory is required. In a classical-memory realization, the stored record is used
to select the system's subsequent evolution.

The question addressed here is different. The classical-versus-quantum-memory
distinction concerns the type of memory needed to realize the reduced dynamics,
whereas the extrinsic-versus-genuine distinction concerns the explanation of a
particular information revival. In particular, the classical record in the
former framework is functionally involved in the realization, since it
determines the subsequent system evolution. By contrast, the inert registers
considered here never participate in the interactions and only reveal
correlations already present \textit{before} the dynamics. The two
classifications are therefore complementary, but should not be identified.

The analysis of Ref.~\cite{buscemi-2025-non-causal-backflow} is formulated at
the level of quantum channels. It uses a maximally entangled reference--system
state and therefore probes the full quantum input state space through the
Choi--Jamio\l{}kowski
representation~\cite{jamiolkowski1972linear,choi1975completely}. This is natural
when the channel itself is the object of study. It is less natural when only a
restricted family of initial states can be prepared. Such restrictions arise,
for example, from conservation or superselection rules, symmetry constraints,
limited experimental control, or classicality. Embedding a classical stochastic
process into a quantum channel also introduces nonclassical inputs that are not
part of the original operational process.

Here we extend the distinction between extrinsic revivals and genuine backflow
to processes with restricted initial preparations. The allowed states of the
system $Q$ are represented by the states obtained by steering a fixed bipartite
probe $\rho_{RQ}$ from a reference $R$~\cite{buscemi-NCPTP}. The probe plays the
role of a restricted-input Choi state: it determines both the accessible
preparations and the reference with respect to which information is measured.
Unlike the usual Choi state, it may be mixed and need not be maximally
entangled.

The use of a general probe introduces new explanations that were not possible in
the channel setting. Besides inert extensions of the environment, a mixed probe
may now admit inert extensions of the initial reference--system side. We call
this the ``system side'' for brevity, although it includes the reference and any
preparation record. Information revivals may therefore admit system-side
explanations, environment-side explanations, or explanations requiring
inert registers on both sides of the initial system--environment cut.

We consider a minimal three-time process (with the initial, intermediate, and
final times denoted by $t_0<t_1<t_2$, respectively)
\[
\rho_{RQ}
\longrightarrow
\sigma_{RQ'}
\longrightarrow
\tau_{RQ''},
\]
generated by two local interactions between the system and an initially
independent environment. The reference system $R$ remains inert throughout the
entire evolution. We then enlarge the interaction model by adjoining additional
inert systems that may initially be correlated either with the reference--system
probe or with the environment, but never participate in the dynamics. A revival
is called extrinsic if retaining suitable additional inert degrees of freedom
restores the monotonic decrease of reference--system information; otherwise, it
is called genuine.

Our main conclusion is that \textit{genuine information backflow requires
quantumness in the initial preparations}. Here quantumness has a precise and
elementary meaning: the preparable states cannot all commute. A system with a
commuting preparable set may still exhibit information revivals, but every such
revival can be explained by correlations already present in suitable inert
degrees of freedom. Such revivals therefore cannot, by themselves, certify that
information has dynamically returned from the active environment.

This conclusion follows from an exact structural relation (which may be of
independent interest) between the probe and the preparations it generates: a
reference--system probe is quantum--classical on $Q$ if and only if its
preparable set is commuting.

We also compare the different explanation classes. The
 explainability on the system- and environment-side is incomparable:
each can occur without the
other. Both, however, imply global explainability, while a global explanation
need not reduce to either one-sided explanation. This gives a V-shaped inclusion
structure rather than a linear hierarchy.

The paper is organized as follows. Section~\ref{sec:formalization} introduces
preparable sets, local interaction models, and their inert extensions.
Section~\ref{sec:revivals-explanations} defines information revivals, their
explanations, and the corresponding explainability classes.
Section~\ref{sec:results} contains our main results:
Theorem~\ref{thm:strict-hierarchy} establishes the hierarchy among system-side,
environment-side, and global explanations; Theorem~\ref{thm:system-screening}
shows that every revival from a separable probe is extrinsic; and
Lemma~\ref{lem:commuting-qc} proves that a probe is quantum--classical on $Q$ if
and only if its preparable set is commuting.
Corollary~\ref{cor:classical-no-backflow} then concludes that genuine backflow
requires noncommuting initial system states.
Theorem~\ref{thm:asymmetric-conditioning} reveals an asymmetry between fixed
system-side and environment-side inert extensions. Finally,
Section~\ref{sec:numerical-illustrations} illustrates with numerical examples
the different explanatory mechanisms.

\section{Formalization}
\label{sec:formalization}

All systems are finite-dimensional. We label systems by $Q,R,E,\dots$ and their
Hilbert spaces by $\mathcal{H}_Q,\mathcal{H}_R,\mathcal{H}_E,\dots$.
For a Hilbert space $\mathcal{H}$, let
$\mathsf{L}(\mathcal{H})$ be the linear operators, $\mathsf{P}
(\mathcal{H})$ the positive semidefinite
operators, and $\mathsf{S}(\mathcal{H})$ the normalized states.
A restricted state set is
denoted by $\mathcal{S}\subseteq\mathsf{S}(\mathcal{H})$.
Quantum channels (that is, completely
positive trace-preserving linear maps) are written as
$\mathcal{N}:\mathsf{L}(\mathcal{H}_A)\to\mathsf{L}(\mathcal{H}_B)$,
or simply $\mathcal{N}:A\to B$. We use the same symbol
for a unitary operator and its unitary channel. Thus $U:QE\to Q'E'$ denotes a
unitary from $\mathcal{H}_Q\otimes\mathcal{H}_E$ to
$\mathcal{H}_{Q'}\otimes\mathcal{H}_{E'}$. We allow the
local dimensions to change during an interaction, provided the total dimensions
agree.

An open evolution of $Q$ is generated by interaction with an environment $E$.
The system starts in an allowed state
$\rho_Q\in\mathcal{S}_Q\subseteq\mathsf{S}(\mathcal{H}_Q)$, while
the environment starts in a fixed state $\omega_E\in\mathsf{S}
(\mathcal{H}_E)$. We then
consider a minimal three-time process with the initial, intermediate, and final
times denoted by $t_0<t_1<t_2$, respectively, generated by two local
interactions between the system and an initially independent environment. The
interactions during the intervals from $t_0$ to $t_1$ and from $t_1$ to $t_2$
are described by unitaries $U:QE\to Q'E'$ and $V:Q'E'\to Q''E''$, respectively.
This general formulation includes, in particular, two successive periods of
evolution under the same time-independent Hamiltonian $H$. When the system and
environment Hilbert spaces are unchanged, and setting $\hbar=1$, these are
$U=e^{-i(t_1-t_0)H}$ and $V=e^{-i(t_2-t_1)H}$, as in the numerical illustrations
of Section~\ref{sec:numerical-illustrations}.

When every input state is allowed, $\mathcal{S}_Q=\mathsf{S}
(\mathcal{H}_Q)$, the evolution is
specified by the two channels (obtained by tracing over the environment degrees
of freedom) from the initial time to $t_1$ and $t_2$:
\begin{equation*}
\mathcal{N}_{t_0\to t_1}:Q\to Q'
\qquad\text{and}\qquad
\mathcal{N}_{t_0\to t_2}:Q\to Q''.
\end{equation*}
The same information can be encoded by the Choi--Jamio\l{}kowski
representation~\cite{jamiolkowski1972linear,choi1975completely,
nielsen_chuang_2010}:
one prepares a reference $R\cong Q$ and the maximally entangled state
$|\Phi^+_{RQ}\rangle:=d^{-1/2}\sum_i|i_R\rangle\otimes|i_Q\rangle$, 
then lets the dynamics act
on $Q$. Ref.~\cite{buscemi-2025-non-causal-backflow} uses this channel-level
description.

Allowing every input state excludes many operational settings. Conservation
laws, superselection rules, symmetries, or limited control may restrict the
allowed inputs to a proper subset $\mathcal{S}_Q\subset\mathsf{S}
(\mathcal{H}_Q)$~\cite{Asymy3,Asymy5}.
In this case, the standard Choi representation is not able to describe such a
restriction. We therefore use the idea of \textit{finitely preparable
sets}~\cite{buscemi-NCPTP}.
A set $\mathcal{S}_Q\subseteq\mathsf{S}(\mathcal{H}_Q)$ is finitely
preparable if there are a finite-dimensional system (i.e., a ``pump'') $A$, with
Hilbert space $\mathcal{H}_A$, and a completely positive map
$\mathcal{F}:A\to Q$, not
necessarily trace preserving, such that
\[
\mathcal{S}_Q=
\left\{
\frac{\mathcal{F}(\rho_A)}
{\operatorname{Tr}[\mathcal{F}(\rho_A)]}
:
\rho_A\in\mathsf{S}(\mathcal{H}_A),\
\operatorname{Tr}[\mathcal{F}(\rho_A)]>0
\right\}.
\]
Thus $\mathcal{S}_Q$ is the normalized image of the completely positive map.

Equivalently~\cite{buscemi-NCPTP}, there are a finite-dimensional reference $R$
and a bipartite state $\rho_{RQ}\in\mathsf{S}
(\mathcal{H}_R\otimes\mathcal{H}_Q)$ whose steered set is
exactly $\mathcal{S}_Q$:
\[
\mathcal{S}_Q=
\left\{
\frac{\operatorname{Tr}_R[(M_R\otimes I_Q)\rho_{RQ}]}
{\operatorname{Tr}[(M_R\otimes I_Q)\rho_{RQ}]}
:
M_R\in\mathsf{P}(\mathcal{H}_R),
\operatorname{Tr}[(M_R\otimes I_Q)\rho_{RQ}]>0
\right\}.
\]
Here $M_R$ is any positive semidefinite operator on $R$, or equivalently a POVM
effect up to positive rescaling.

Every preparable set is convex,
simply because the set $\mathsf{P}(\mathcal{H}_R)$ is convex.
This includes state polytopes, simplices of commuting states, and superselection
sectors. We shall call a preparable set \emph{classical} when all its states
commute, equivalently when they are simultaneously diagonal in one basis. A
preparable set is \emph{noncommuting} when this is not the case. A classical
preparable set is a convex subset of a finite classical simplex, but need not
itself be a simplex or a polytope, as the following short example shows.

\begin{definition}{Example}(A commuting preparable set that is not a polytope)
\label{ex:commuting-non-polytope}
Let $Q$ be three-dimensional with classical basis $\{\ket{1},\ket{2},\ket{3}\}$
and let the pump $A$ be a qubit. Consider the channel
\[
\mathcal{F}(\rho)
=
\sum_{i=1}^3 \operatorname{Tr}[E_i\rho]\,\ket{i}\!\bra{i},
\]
where
\[
E_1=\frac{1}{4}(I+\sigma_x),\qquad
E_2=\frac{1}{4}(I+\sigma_y),\qquad
E_3=I-E_1-E_2
=
\frac{1}{2}I-\frac{1}{4}(\sigma_x+\sigma_y).
\]
The effects are positive and sum to $I$: in particular, the eigenvalues of $E_3$
are $(2\pm\sqrt{2})/4$.

For a qubit state
\[
\rho=\frac{1}{2}(I+x\sigma_x+y\sigma_y+z\sigma_z),
\qquad x^2+y^2+z^2\leqslant 1,
\]
the output probabilities depend only on $x$ and $y$:
\[
p_1=\frac{1+x}{4},\qquad
p_2=\frac{1+y}{4},\qquad
p_3=1-p_1-p_2=\frac{2-x-y}{4}.
\]
Although $z$ does not appear in the output, every pair $(x,y)$ with
$x^2+y^2\leqslant1$ is realized, for instance by choosing $z=0$.
Hence the preparable
set is
\[
\mathcal{S}_Q
=
\left\{
\operatorname{diag}(p_1,p_2,p_3):
(4p_1-1)^2+(4p_2-1)^2\leqslant 1,\quad
p_3=1-p_1-p_2
\right\}.
\]
All these states commute, but the disk constraint gives a curved boundary and
infinitely many extreme points. Thus $\mathcal{S}_Q$ is not a polytope.
\nobreak\hfill\ensuremath{\square}\end{definition}

\begin{definition}{Remark}(The probe is part of the data)
\label{rem:probe}
A preparable set $\mathcal{S}_Q$ need not determine a unique bipartite
probe. Different
states $\rho_{RQ}$ may generate the same steered set, and we therefore fix one
probe as part of the process specification. All information quantities and
explainability statements below are relative to this choice. One structural
property is nevertheless determined entirely by the preparable set:
Lemma~\ref{lem:commuting-qc} shows that a probe is quantum--classical on $Q$ if
and only if its steered set is commuting. Thus every probe representing a
commuting preparable set is quantum--classical and separable. By contrast, we
are aware of a counterexample showing that the same noncommuting preparable set
can be generated by both a separable and an entangled probe. Thus separability,
unlike quantum--classicality on $Q$, is not determined by the preparable set.
\nobreak\hfill\ensuremath{\square}\end{definition}

\subsection{Local interaction models and their inert extensions}

The operational data are the three reference--system states
\begin{equation}
\rho_{RQ}
\longrightarrow
\sigma_{RQ'}
\longrightarrow
\tau_{RQ''}\;,
\label{eq:three-time-rs-process}
\end{equation}
where $R$ is the reference system associated with the preparable set of initial
states of $Q$, while $Q,Q',Q''$ denote the system at the three times
$t_0,t_1,t_2$, respectively. Since the reference remains inert throughout the
process, any admissible triple necessarily satisfies
\begin{equation}
\rho_R=\sigma_R=\tau_R.
\label{eq:reference-marginal-fixed}
\end{equation}
The same steering operation represents the preparation at all three times. More
precisely, for any effect $M_R$ with nonzero steering probability, applying
$M_R$ to $\rho_{RQ}$, $\sigma_{RQ'}$, and $\tau_{RQ''}$ gives the selected
initial state and its intermediate and final outputs.
Equation~\eqref{eq:reference-marginal-fixed} ensures that the steering
probability is the same at all three times. The triple therefore represents the
dynamics on the restricted set of preparations, not merely three unrelated
reference--system states.

The states in Eq.~\eqref{eq:three-time-rs-process} however do not specify which
environmental degrees of freedom may be responsible for that evolution.

\begin{definition}{Definition}(Local interaction model)\label{def:LIM}
A triple $(\rho_{RQ},\sigma_{RQ'},\tau_{RQ''})$ is \emph{admissible} if it
admits a \emph{local interaction model} (LIM): an environment $E$, an initial
state $\omega_E\in\mathsf{S}(\mathcal{H}_E)$, and unitary interactions
\begin{equation}
U:QE\to Q'E',
\qquad
V:Q'E'\to Q''E'',
\label{eq:LIM-unitaries}
\end{equation}
such that, defining
\begin{align}
\sigma_{RQ'E'}
&:=
(I_R\otimes U_{QE})\bigl(\rho_{RQ}\otimes\omega_E\bigr)(I_R\otimes
U_{QE}^\dagger),
\label{eq:LIM-first-step}
\\
\tau_{RQ''E''}
&:=
(I_R\otimes V_{Q'E'})\sigma_{RQ'E'}(I_R\otimes V_{Q'E'}^\dagger),
\label{eq:LIM-second-step}
\end{align}
one has
\begin{equation}
\operatorname{Tr}_{E'}[\sigma_{RQ'E'}]=\sigma_{RQ'}\;,
\qquad
\operatorname{Tr}_{E''}[\tau_{RQ''E''}]=\tau_{RQ''}\;.
\label{eq:LIM-marginals}
\end{equation}
A LIM satisfying Eq.~\eqref{eq:LIM-marginals} is said to be \emph{compatible}
with the triple.
\nobreak\hfill\ensuremath{\square}\end{definition}

The initial environment is independent of $RQ$, as in the standard open-system
setting~\cite{Breuer-Petruccione,vacchini-2024-open-quantum-textbook}. The
unitaries may change the local dimensions, provided
$\mathcal{H}_Q\otimes\mathcal{H}_E\cong\mathcal{H}_{Q'}\otimes\mathcal{
H}_{E'}\cong\mathcal{H}_{Q''}\otimes\mathcal{H}_{E''}$.
Different compatible LIMs may yield the same reduced snapshots.

To distinguish a revival from its explanation, we enlarge a LIM by systems that
never take part in the dynamics. These systems may extend, purify, or label
correlations already present at $t_0$. Because no interaction acts on them, they
cannot themselves carry information away and later return it.

\begin{definition}{Definition}(Extended local interaction model)
\label{def:eLIM}
Let $(\rho_{RQ},\sigma_{RQ'},\tau_{RQ''})$ have a compatible LIM as in
Definition~\ref{def:LIM}. An \emph{extended LIM} adds finite-dimensional inert
systems $S$ and $F$ and an initial product state across the enlarged
system--environment cut,
\begin{equation}
\widetilde{\rho}_{SRQ}\otimes\widetilde{\omega}_{EF}
\in
\mathsf{S}
(\mathcal{H}_S\otimes\mathcal{H}_R\otimes\mathcal{H}_Q\otimes\mathcal{H
}_E\otimes\mathcal{H}_F)
\end{equation}
such that
\begin{equation}
\operatorname{Tr}_{S}\widetilde{\rho}_{SRQ}
=
\rho_{RQ},
\qquad
\operatorname{Tr}_F\widetilde{\omega}_{EF}
=
\omega_E.
\label{eq:eLIM-initial-marginals}
\end{equation}
The corresponding extended evolution is defined by
\begin{align}
\widetilde{\sigma}_{SRQ'E'F}
&:=
(I_{SR}\otimes U_{QE}\otimes I_F)
(\widetilde{\rho}_{SRQ}\otimes\widetilde{\omega}_{EF})
(I_{SR}\otimes U_{QE}^{\dagger}\otimes I_F),
\label{eq:eLIM-first-step}
\\
\widetilde{\tau}_{SRQ''E''F}
&:=
(I_{SR}\otimes V_{Q'E'}\otimes I_F)
\widetilde{\sigma}_{SRQ'E'F}
(I_{SR}\otimes V_{Q'E'}^{\dagger}\otimes I_F).
\label{eq:eLIM-second-step}
\end{align}
The systems $S$ and $F$, like $R$, are inert throughout the evolution. The
register $S$ extends the system side, namely the initial reference--system or
preparation side, while $F$ extends the environment side. An extended LIM
satisfying these marginal conditions is said to be compatible with the same
reduced triple.
\nobreak\hfill\ensuremath{\square}\end{definition}

\begin{definition}{Remark}
Definition~\ref{def:eLIM} deliberately excludes a general extension
\[
\widetilde{\rho}_{SRQEF}
\quad\text{with}\quad
\operatorname{Tr}_{SF}\bigl[\widetilde{\rho}_{SRQEF}\bigr]
=
\rho_{RQ}\otimes\omega_E .
\]
This is because, allowing arbitrary cross-cut correlations would mix
pre-existing and dynamically generated correlations and would make the backflow
question ambiguous.
\nobreak\hfill\ensuremath{\square}\end{definition}

Note that, by definition, tracing out the inert registers recovers the original
LIM:
\begin{equation}
\operatorname{Tr}_{SF}\bigl[\widetilde{\sigma}_{SRQ'E'F}\bigr]
=
\sigma_{RQ'E'},
\qquad
\operatorname{Tr}_{SF}\bigl[\widetilde{\tau}_{SRQ''E''F}\bigr]
=
\tau_{RQ''E''}.
\label{eq:eLIM-recovers-LIM}
\end{equation}
Thus, an extended LIM leaves the operational process unchanged and only enlarges
the set of possible explanations. Here an inert extension is called
\emph{trivial} if the added system is one-dimensional or is initially
uncorrelated with all other systems; since it remains inert, it can then be
omitted. Accordingly, if $F$ is trivial, only the system side is extended,
whereas if $S$ is trivial, only the environment side is extended.

\section{Information revivals and their explanations}
\label{sec:revivals-explanations}

For a bipartite state $\xi_{AB}$, the quantum mutual information is
\begin{equation}
I(A;B)_{\xi}
:=
H(A)_{\xi}
+
H(B)_{\xi}
-
H(AB)_{\xi},
\label{eq:mutual-information}
\end{equation}
where $H(\xi):=-\operatorname{Tr}[\xi\log\xi]$ is the von Neumann 
entropy. All logarithms are to base two, so information quantities are 
measured in bits.
For a tripartite state $\xi_{ABC}$, the quantum conditional mutual 
information
(QCMI) is
\begin{equation}
I(A;B|C)_{\xi}
:=
H(AC)_{\xi}
+
H(BC)_{\xi}
-
H(C)_{\xi}
-
H(ABC)_{\xi}.
\label{eq:conditional-mutual-information}
\end{equation}

The reference system also gives a precise meaning to the information carried by
$Q$: we measure it through the correlations between $R$ and $Q$.  From now on,
subscripts $0$, $1$, and $2$ on information quantities indicate evaluation on
the initial, intermediate, and final configurations, respectively.

\begin{definition}{Definition}(Information
revival~\cite{luo-2012-non-markovianity-via-correlations})
\label{def:information-revival}
An admissible three-time process
\begin{equation}
\rho_{RQ}
\longrightarrow
\sigma_{RQ'}
\longrightarrow
\tau_{RQ''}
\end{equation}
exhibits an \emph{information revival} between $t_1$ and $t_2$ if
\begin{equation}
I(R;Q'')_2
>
I(R;Q')_1.
\label{eq:information-revival}
\end{equation}
\nobreak\hfill\ensuremath{\square}\end{definition}

Admissibility already implies $I(R;Q')_1\leqslant I(R;Q)_0$ 
and $I(R;Q'')_2\leqslant
I(R;Q)_0$, by data processing from the initial state to each later reduced
state. Therefore, Eq.~\eqref{eq:information-revival} necessarily implies
$I(R;Q')_1<I(R;Q)_0$: some information is first lost from the system and later
recovered.

A revival rules out a local channel from $Q'$ to $Q''$. Indeed, if a channel
$\mathcal{D}:Q'\to Q''$ satisfied
\begin{equation}
\tau_{RQ''}
=
(\operatorname{id}_R\otimes\mathcal{D})(\sigma_{RQ'}),
\label{eq:local-intermediate-channel}
\end{equation}
then the data-processing inequality~\cite{nielsen_chuang_2010} would give
\begin{equation}
I(R;Q'')_2
\leqslant
I(R;Q')_1,
\label{eq:intermediate-DPI}
\end{equation}
contrary to Eq.~\eqref{eq:information-revival}. Thus the reduced intermediate
step cannot be a deterministic operation on $Q'$ alone. We now ask which
additional degrees of freedom restore monotonicity.

\begin{definition}{Definition}(Explanation of a revival)
\label{def:revival-explanation}
Consider a revival as in Eq.~\eqref{eq:information-revival} and fix one
compatible LIM or extended LIM. Retain the same physical degrees of freedom at
the two successive times, denoting them by $X'$ at $t_1$ and by $X''$ at $t_2$.
Inert registers keep the same labels, while the active environment is denoted by
$E'$ and $E''$. The corresponding marginals satisfy
\begin{equation}
\operatorname{Tr}_{X'}[\widetilde{\sigma}_{RQ'X'}]
=
\sigma_{RQ'},
\qquad
\operatorname{Tr}_{X''}[\widetilde{\tau}_{RQ''X''}]
=
\tau_{RQ''}.
\label{eq:explanation-marginals}
\end{equation}
The retained degrees of freedom \emph{explain} the revival if
\begin{equation}
I(R;Q'X')_1
\geqslant
I(R;Q''X'')_2.
\label{eq:general-explanation}
\end{equation}
\nobreak\hfill\ensuremath{\square}\end{definition}

The systems $X'$ and $X''$ are not chosen independently. They are the same
retained part of one compatible model at the two times. For example, retaining
$S$ gives $X'=X''=S$, retaining $SF$ gives $X'=X''=SF$, and retaining the active
environment together with $F$ gives $X'=E'F$ and $X''=E''F$. Thus a revival is a
property of the reduced triple $(\rho_{RQ},\sigma_{RQ'},\tau_{RQ''})$, whereas
an explanation is a property of a compatible model and a specified retained part
of that model.

Unless explicitly stated otherwise, explainability below is
\textit{existential}: the compatible model and the retained inert registers are
part of the explanation whose existence is asserted.

As a baseline, the active environment always provides an explanation.

\begin{theorem}{Proposition}(The active environment as a universal explanation)
\label{prop:environment-universal-explanation}
Every revival admitted by a LIM is explained by retaining the active
environment, that is, by taking $X'=E'$ and $X''=E''$ in
Definition~\ref{def:revival-explanation}.
\nobreak\hfill\ensuremath{\square}\end{theorem}

\noindent{\it Proof.}
The second interaction is unitary on $Q'E'$ and leaves $R$ unchanged. Therefore
\begin{equation}
I(R;Q'E')_1
=
I(R;Q''E'')_2.
\label{eq:total-information-conservation}
\end{equation}
Equation~\eqref{eq:general-explanation} holds with equality.
\ensuremath{\blacksquare}%

This is the standard backflow explanation: information about $R$ that is absent
from $Q'$ is accessible when the active environment $E'$ is retained alongside
$Q'$, and later becomes accessible from $Q''$ alone. A genuine-backflow test
asks whether this active explanation is necessary. It may not be so. An inert
register can reveal correlations already present at $t_0$ and thereby remove the
apparent increase without serving as a dynamical
memory~\cite{buscemi-2025-non-causal-backflow}.

\begin{definition}{Definition}(Extrinsic revivals and genuine backflows)
\label{def:extrinsic-genuine-revival}
An information revival is \emph{extrinsic} if it can be explained without
retaining the active environment. Explicitly, there must exist a compatible
extended LIM and a retained inert part $X\in\{S,F,SF\}$ such that
\begin{equation}
I(R;Q'X)_1
\geqslant
I(R;Q''X)_2.
\label{eq:extrinsic-quantifier}
\end{equation}
The revival is \emph{genuine}, or a \emph{genuine backflow}, if no inert
explanation of this kind exists. Equivalently, for every compatible extended LIM
and every retained inert part $X\in\{S,F,SF\}$,
\begin{equation}
I(R;Q'X)_1
<
I(R;Q''X)_2.
\label{eq:genuine-quantifier}
\end{equation}
\nobreak\hfill\ensuremath{\square}\end{definition}

The quantifiers are asymmetric. To show that a revival is extrinsic, it is
enough to find one compatible extended LIM and one retained inert part
satisfying Eq.~\eqref{eq:extrinsic-quantifier}. To show that it is genuine, one
must rule out \textit{every} such inert explanation: see
Ref.~\cite{buscemi-2025-non-causal-backflow} for further details on this point.

We next distinguish where the explanatory inert registers lie.

\begin{definition}{Definition}(Inert explainability classes)
\label{def:explainability-classes}
Consider a revival as in Eq.~\eqref{eq:information-revival}.

It is \emph{environment-side explainable} if there exists a compatible extended
LIM with trivial $S$ such that
\begin{equation}
I(R;Q'F)_1
\geqslant
I(R;Q''F)_2.
\label{eq:F_MI}
\end{equation}

It is \emph{system-side explainable} if there exists a compatible extended LIM
with trivial $F$ such that
\begin{equation}
I(R;Q'S)_1
\geqslant
I(R;Q''S)_2.
\label{eq:S_MI}
\end{equation}

It is \emph{globally explainable} if there exists a compatible extended LIM such
that
\begin{equation}
I(R;Q'SF)_1
\geqslant
I(R;Q''SF)_2.
\label{eq:SF_MI}
\end{equation}
Here either $S$ or $F$ may be trivial.
\nobreak\hfill\ensuremath{\square}\end{definition}

By definition, then, environment-side explanations and system-side explanations
are special cases of global explanations. Hence
\begin{equation}
\text{extrinsic}
\quad\Longleftrightarrow\quad
\text{globally explainable}.
\label{eq:extrinsic-global-equivalence}
\end{equation}
As shown below in Theorem~\ref{thm:strict-hierarchy}, a joint $SF$ explanation
need not remain explanatory after either $S$ or $F$ is discarded. 
Figures~\ref{fig:system-world-compensation},
\ref{fig:environment-world-compensation}, and \ref{fig:two-sided-compensation}
in Section~\ref{sec:numerical-illustrations} numerically illustrate cases in
which the displayed system-side, environment-side, and global extensions,
respectively, explain the observed revival. Each case uses a fixed extended LIM,
and the other two displayed extensions fail to
explain the reviva.

\begin{definition}{Remark}
\label{rem:QCMI-explainability}
For any fixed compatible extended LIM, the explanation inequalities can be
written using quantum conditional mutual information. Unitary invariance and the
chain rule show that retaining $F$ explains the revival if and only if
\begin{equation}
I(R;E''|Q''F)_2
\geqslant
I(R;E'|Q'F)_1,
\end{equation}
retaining $S$ explains it if and only if
\begin{equation}
I(R;E''|Q''S)_2
\geqslant
I(R;E'|Q'S)_1,
\end{equation}
and retaining $SF$ explains it if and only if
\begin{equation}
I(R;E''|Q''SF)_2
\geqslant
I(R;E'|Q'SF)_1.
\end{equation}
\nobreak\hfill\ensuremath{\square}\end{definition}

\section{Results}
\label{sec:results}

Throughout this section, $(\rho_{RQ},\sigma_{RQ'},\tau_{RQ''})$ is an admissible
three-time process. We first compare the three existential explainability
classes, then prove the no-go result for separable probes and the exact
equivalence between commuting preparable sets and quantum--classical probes, and
finally study exact Markov screening condition in one fixed extended model.

\subsection{A V-shaped hierarchy of explanations}

The one-sided inclusions established above are both strict.

\begin{theorem}{Theorem}(V-shaped hierarchy of explanations)
\label{thm:strict-hierarchy}
System-side and environment-side explainability are incomparable:
\begin{align}
\text{System} &\;\not\Rightarrow\; \text{Environment},
&
\text{Environment} &\;\not\Rightarrow\; \text{System}.
\end{align}
The above, in turn, implies that, in general, 
\begin{align}
\text{Global} &\;\not\Rightarrow\; \text{System},
&
\text{Global} &\;\not\Rightarrow\; \text{Environment},
\label{eq:global-not-imply}
\end{align}
because, by definition, environment-side explanations and system-side
explanations are special cases of global explanations.
\nobreak\hfill\ensuremath{\square}\end{theorem}

\noindent{\it Proof.}
We give two examples: the first is environment-side explainable but not
system-side explainable; the second is system-side explainable but not
environment-side explainable.

For the first example, Ref.~\cite{buscemi-2025-non-causal-backflow} already
provides an environment-side explainable revival that starts from a pure (in
fact, maximally entangled) $\rho_{RQ}$. For any such process, every system-side
inert extension is a product extension,
\[
\widetilde{\rho}_{RQS}
=
\rho_{RQ} \otimes \zeta_S,
\]
for some state $\zeta_S$. Because $S$ is inert, it stays product with the
reduced reference--system state. Hence
\[
I(R;Q'S)_1 = I(R;Q')_1,
\qquad
I(R;Q''S)_2 = I(R;Q'')_2.
\]
The reduced revival therefore remains a revival after adding any such $S$: that
is, no system-side explanation exists.

For the second example, choose a correlated separable state $\rho_{RQ}$, so that
$I(R;Q)_0>0$. It is then easy to see that a register $S$ exists such that
\[
I(R;Q|S)_0=0.
\]
Now, let us define
\[
\sigma_{RQ'}
=
\rho_R\otimes\ket{0}\!\bra{0}_{Q'},
\qquad
\tau_{RQ''}
=
\rho_{RQ}.
\]
This triple has a compatible LIM: take $E\cong Q$ in the pure state
$\ket{0}\!\bra{0}_E$, swap $Q$ and $E$ at the first interaction, and swap them
back at the second. The reduced process has the revival
\[
I(R;Q')_1
=
0
<
I(R;Q'')_2
=
I(R;Q)_0.
\]
The register $S$ explains it because
\[
I(R;Q'S)_1
=
I(R;S)_0
=
I(R;QS)_0
=
I(R;Q''S)_2.
\]

We now exclude all environment-side explanations. Consider the inert register
$F$ in any compatible extended LIM. Initial factorization gives $I(R;F)_0=0$,
and the first interaction leaves the $RF$ marginal unchanged. Since the
prescribed state of $Q'$ is pure, $Q'$ is product with every extension of
$\sigma_{RQ'}$. Therefore
\[
I(R;Q'F)_1
=
I(R;F)_1
=
0.
\]
At the final time, data processing under the partial trace over $F$ gives
\[
I(R;Q''F)_2
\geqslant
I(R;Q'')_2
=
I(R;Q)_0
>
0.
\]
Thus no environment-side explanation exists.
\ensuremath{\blacksquare}%

Thus a revival may appear a genuine backflow under a one-sided explanation,
while becoming extrinsic as soon as global explanations are allowed.

\begin{definition}{Remark}
Purity limits possible inert extensions inside a fixed model. For example, if
the initial environment state of a chosen LIM is pure, every environment-side
extension of that state is product, so any nontrivial inert explanation within
that fixed model must come from the system side. If $\rho_{RQ}$ is also pure,
every system-side extension is product as well. Hence, within a model that is
globally pure (a sort of perfectly isolated universe),
\[
I(R;Q'SF)_1=I(R;Q')_1,
\qquad
I(R;Q''SF)_2=I(R;Q'')_2.
\]
Any reduced revival therefore is genuine. However, this fixed-model statement
does not rule out an inert explanation arising from a \textit{different}
compatible extended LIM.
\nobreak\hfill\ensuremath{\square}\end{definition}

\subsection{Quantumness of the initial preparations is necessary for genuine
backflow}
\label{subsec:initial-entanglement}

We now show that genuine backflow requires quantumness in the initial
preparations. At the operational level, quantumness has a precise and elementary
meaning: the preparable system states cannot all commute. At the probe level,
the corresponding necessary condition is entanglement of the reference--system
state. The probe-level result follows directly from separability. The label of
any separable decomposition of $\rho_{RQ}$ can be stored in a classical inert
system $S$, producing an extension such that
\[
I(R;Q|S)_0=0.
\]
The initial correlations between $R$ and $Q$ are then screened by $S$, and the
same extension provides a system-side explanation of every subsequent revival,
independently of the compatible interaction model.

\begin{theorem}{Theorem}(Every revival from a separable probe is extrinsic)
\label{thm:system-screening}
Consider an admissible triple $(\rho_{RQ},\sigma_{RQ'},\tau_{RQ''})$. If
$\rho_{RQ}$ is separable, then any information revival of the reduced process is
system-side explainable. In particular, it is globally explainable, and hence
extrinsic.
\nobreak\hfill\ensuremath{\square}\end{theorem}

\noindent{\it Proof.}
Write the separable state as
\begin{equation}
\rho_{RQ}
=
\sum_x p_x \rho_R^x \otimes \rho_Q^x .
\end{equation}
Add a classical inert register $S$ that stores the decomposition label:
\begin{equation}
\widetilde{\rho}_{SRQ}
=
\sum_x p_x
\ket{x}\!\bra{x}_S
\otimes
\rho_R^x
\otimes
\rho_Q^x .
\end{equation}
This state extends $\rho_{RQ}$. Conditioned on $S=x$, the state of $RQ$ is
product, so
\begin{equation}
I(R;Q|S)_0
=
\sum_x p_x I(R;Q)_{\rho_R^x\otimes\rho_Q^x}
=
0 ,
\label{eq:separable-screening}
\end{equation}
and $S$ screens the initial $R$--$Q$ correlations.

Fix any LIM compatible with the reduced process. With trivial $F$, its initial
extended state is
\begin{equation}
\widetilde{\rho}_{SRQ}\otimes\omega_E ,
\end{equation}
and its later states are $\widetilde{\sigma}_{SRQ'E'}$ and
$\widetilde{\tau}_{SRQ''E''}$.

The first interaction is unitary on $QE$ and leaves $R$ and $S$ inert, hence
\begin{equation}
I(R;QE|S)_0
=
I(R;Q'E'|S)_1 .
\label{eq:screening-unitary-invariance}
\end{equation}
At $t_0$, the environment is independent of $SRQ$, so
\begin{equation}
I(R;QE|S)_0
=
I(R;Q|S)_0 .
\label{eq:screening-environment-reduction}
\end{equation}
Together with Eq.~\eqref{eq:separable-screening}, these identities give
\begin{equation}
I(R;Q'E'|S)_1=0 .
\end{equation}
By the chain rule,
\begin{equation}
I(R;Q'E'|S)_1
=
I(R;Q'|S)_1
+
I(R;E'|Q'S)_1 .
\end{equation}
Both terms are nonnegative, and therefore
\begin{equation}
I(R;E'|Q'S)_1=0 .
\label{eq:separable-decoupling}
\end{equation}
Nonnegativity also gives
\begin{equation}
I(R;E''|Q''S)_2
\geqslant
I(R;E'|Q'S)_1 .
\end{equation}
The QCMI criterion in Remark~\ref{rem:QCMI-explainability} now shows that the
revival is system-side explainable. Taking $F$ to be trivial turns the same
system-side explanation into a global explanation, so the revival is extrinsic.
\ensuremath{\blacksquare}%

The converse viewpoint is related to squashed
entanglement~\cite{christandl-winter-2004-squashed}, defined, up to a
conventional factor, by the infimum of $I(R;Q|S)$ over all extensions $S$. No
such optimization is needed for the implication used here: a separable
decomposition directly supplies a finite classical extension satisfying
$I(R;Q|S)_0=0$. Although this analogy does not establish a formal complexity
result, the optimization over arbitrary extensions suggests that finding
explanatory extensions in general extended LIMs may be computationally
demanding. Figure~\ref{fig:entanglement-revival-landscape} in
Section~\ref{sec:numerical-illustrations} numerically illustrates
Theorem~\ref{thm:system-screening}. The figure also shows how noncommuting
initial preparations, although necessary, are not sufficient for genuine
backflow, as will be discussed below.

We now pass from the probe-level result to the preparable set itself. The
following lemma gives an exact correspondence between two notions of
classicality: a probe is quantum--classical on $Q$ if and only if all the states
that it prepares on $Q$ commute. Thus, commuting preparable sets do not merely
admit a convenient quantum--classical representation; every probe generating
such a set must be quantum--classical.

\begin{theorem}{Lemma}(Commuting preparable sets characterize
quantum--classical probes)
\label{lem:commuting-qc}
Let $\rho_{RQ}$ be a probe. Its steered set $\mathcal{S}_Q$ consists of mutually
commuting states if and only if, in some orthonormal basis $\{\ket{i}_Q\}$,
\begin{equation}
\rho_{RQ}
=
\sum_i p_i\,\tau_R^i\otimes\ket{i}\!\bra{i}_Q,
\label{eq:qc-probe}
\end{equation}
where $\{p_i\}$ is a probability distribution and 
$\tau_R^i\in\mathsf{S}(\mathcal{H}_R)$ for
every $i$ with $p_i>0$. In particular, $\rho_{RQ}$ is separable.
\nobreak\hfill\ensuremath{\square}\end{theorem}

\noindent{\it Proof.}
Suppose first that Eq.~\eqref{eq:qc-probe} holds. For every 
$M_R\in\mathsf{P}(\mathcal{H}_R)$,
the corresponding unnormalized steered operator is
\begin{equation}
\operatorname{Tr}_R\!\left[(M_R\otimes I_Q)\rho_{RQ}\right]
=
\sum_i p_i\operatorname{Tr}[M_R\tau_R^i]\ket{i}\!\bra{i}_Q.
\end{equation}
Every normalized steered state is therefore diagonal in the basis
$\{\ket{i}_Q\}$, so the states in $\mathcal{S}_Q$ mutually commute.

Conversely, suppose that the states in $\mathcal{S}_Q$ mutually 
commute. Since they are
Hermitian operators on a finite-dimensional space, they are simultaneously
diagonalizable. Fix an orthonormal basis $\{\ket{i}_Q\}$ in which every state in
$\mathcal{S}_Q$ is diagonal. For $M_R\in\mathsf{P}(\mathcal{H}_R)$, define
\begin{equation}
\widetilde{\sigma}_Q(M_R)
:=
\operatorname{Tr}_R\!\left[(M_R\otimes I_Q)\rho_{RQ}\right].
\end{equation}
This operator is non-negative. If $\operatorname{Tr}[M_R\rho_R]>0$, 
its normalization belongs
to $\mathcal{S}_Q$, so $\widetilde{\sigma}_Q(M_R)$ is diagonal. If 
$\operatorname{Tr}[M_R\rho_R]=0$,
then $\widetilde{\sigma}_Q(M_R)$ is non-negative and has vanishing trace, and
hence is zero. Thus, $\widetilde{\sigma}_Q(M_R)$ is diagonal for every
$M_R\in\mathsf{P}(\mathcal{H}_R)$.

Positive semidefinite operators span the whole space $\mathsf{L}
(\mathcal{H}_R)$. Since the map
\begin{equation}
X\longmapsto\operatorname{Tr}_R\!\left[(X\otimes I_Q)\rho_{RQ}\right]
\end{equation}
is linear and the diagonal operators form a complex linear subspace, its image
is diagonal for every $X\in\mathsf{L}(\mathcal{H}_R)$. Writing
\begin{equation}
\rho_{RQ}
=
\sum_{i,j}A_{ij}\otimes\ket{i}\!\bra{j}_Q,
\qquad
A_{ij}:=(I_R\otimes\bra{i}_Q)\rho_{RQ}(I_R\otimes\ket{j}_Q),
\end{equation}
we obtain, for $i\neq j$,
\begin{equation}
0
=
\bra{i}\operatorname{Tr}_R\!\left[(X\otimes I_Q)\rho_{RQ}\right]\ket{j}
=
\operatorname{Tr}[XA_{ij}]
\end{equation}
for every $X\in\mathsf{L}(\mathcal{H}_R)$. This is equivalent to 
$A_{ij}=0$ whenever $i\neq j$.
Hence
\begin{equation}
\rho_{RQ}
=
\sum_i A_{ii}\otimes\ket{i}\!\bra{i}_Q.
\end{equation}
Each $A_{ii}$ is positive and $\sum_i\operatorname{Tr}[A_{ii}]=1$. 
Setting $p_i:=\operatorname{Tr}[A_{ii}]$
and $\tau_R^i:=A_{ii}/p_i$ whenever $p_i>0$, and omitting the terms with
$p_i=0$, gives Eq.~\eqref{eq:qc-probe}. This form is manifestly separable.
\ensuremath{\blacksquare}%

\begin{theorem}{Corollary}(No genuine backflow from classical preparations)
\label{cor:classical-no-backflow}
Let the preparable set $\mathcal{S}_Q$ consist of mutually commuting 
states. Then no
information revival of the reduced process is a genuine backflow: every revival
is system-side explainable, and hence extrinsic.
\nobreak\hfill\ensuremath{\square}\end{theorem}

\noindent{\it Proof.}
For any probe generating $\mathcal{S}_Q$, Lemma~\ref{lem:commuting-qc} 
implies that the
probe is quantum--classical, and hence separable.
Theorem~\ref{thm:system-screening} therefore applies.
\ensuremath{\blacksquare}%

Summarizing, we thus obtain, under the global-explanation convention,
\begin{equation}
\begin{array}{r@{\quad}c@{\quad}c}
\text{genuine backflow}
& \Longrightarrow &
\rho_{RQ}\text{ entangled}
\\[0.2ex]
&& \Downarrow
\\[0.2ex]
\mathcal{S}_Q\text{ noncommuting}
& \Longleftrightarrow &
\substack{
\rho_{RQ}\text{ not quantum--classical}\\[-0.2ex]
\text{on }Q
}
\end{array}
\end{equation}
Thus, noncommutativity of the preparable system states is necessary, although
not sufficient, for genuine backflow.

\begin{definition}{Remark}
\label{rem:initial-entanglement-interpretation}
The no-go result hinges upon the assumption that global explanations are
acceptable. A different operational convention may forbid system-side extensions
because, e.g., the preparation is taken to be fully specified or because $Q$
already includes all controlled degrees of freedom. In such a case, a separable
probe could still support a revival that has no allowed environment-side
explanation.
\nobreak\hfill\ensuremath{\square}\end{definition}

\subsection{An asymmetry under fixed inert extensions}
\label{subsec:asymmetric-inert-conditioning}

The explainability classes above are existential and revival-specific: the
compatible extended LIM and the retained inert registers may be chosen for the
revival under consideration. We now fix one extended LIM, with its given $S$ and
$F$, and impose a stronger condition at the intermediate time.

We say that $S$ gives \emph{exact system-side Markov screening} at $t_1$ if
\begin{equation}
I(R;E'|Q'S)_1=0,
\label{eq:system-side-markov-screening}
\end{equation}
and that $F$ gives \emph{exact environment-side Markov screening} if
\begin{equation}
I(R;E'|Q'F)_1=0.
\label{eq:environment-side-markov-screening}
\end{equation}
Vanishing conditional mutual information is the quantum Markov condition: once
$Q'$ and the retained register are known, $E'$ contains no further information
about $R$~\cite{hayden2004structure}. This is stronger than explaining one
selected revival. If, for example, Eq.~\eqref{eq:system-side-markov-screening}
holds, then
\[
I(R;E''|Q''S)_2\geqslant0=I(R;E'|Q'S)_1
\]
for every later unitary on $Q'E'$. Hence the same $S$ explains every later
revival from $t_1$.

We ask whether exact screening survives when the already fixed register on the
other side is also retained. The answer is asymmetric.

\begin{theorem}{Theorem}(Asymmetric stability of exact Markov screening)
\label{thm:asymmetric-conditioning}
For every compatible extended LIM, the intermediate state satisfies
\begin{equation}
I(R;E'|Q'SF)_1
=
I(R;E'|Q'S)_1
-
I(R;F|Q'S)_1,
\label{eq:system-to-two-sided-identity}
\end{equation}
and
\begin{equation}
I(R;E'|Q'SF)_1
=
I(R;E'|Q'F)_1
+
I(R;S|Q)_0
-
I(R;S|Q'F)_1.
\label{eq:environment-to-two-sided-identity}
\end{equation}
Consequently,
\begin{equation}
I(R;E'|Q'S)_1=0
\quad\Longrightarrow\quad
I(R;E'|Q'SF)_1=0,
\label{eq:system-decoupling-lifts}
\end{equation}
whereas $I(R;E'|Q'F)_1=0$ does not, in general, imply $I(R;E'|Q'SF)_1=0$ for the
same fixed extension.
\nobreak\hfill\ensuremath{\square}\end{theorem}

The first identity, Eq.~\eqref{eq:system-to-two-sided-identity}, has a definite
sign: adding $F$ can only decrease the system-side screening quantity. The
second, Eq.~\eqref{eq:environment-to-two-sided-identity}, contains the residual
difference $I(R;S|Q)_0-I(R;S|Q'F)_1$, which has no fixed sign. This is the
source of the asymmetry.

\noindent{\it Proof.}
Expand $I(R;E'F|Q'S)_1$ in the two possible ways:
\begin{align}
I(R;F|Q'S)_1
+
I(R;E'|Q'SF)_1
&=
I(R;E'|Q'S)_1
+
I(R;F|Q'E'S)_1.
\label{eq:system-chain-comparison}
\end{align}
Unitary invariance under the first interaction gives
\[
I(R;F|Q'E'S)_1=I(R;F|QES)_0.
\]
The initial state factorizes across the enlarged system--environment cut.
Therefore
\[
I(R;F|QES)_0=0.
\]
This proves Eq.~\eqref{eq:system-to-two-sided-identity}. Nonnegativity then
gives
\[
0\leqslant I(R;E'|Q'SF)_1\leqslant I(R;E'|Q'S)_1,
\]
which proves Eq.~\eqref{eq:system-decoupling-lifts}.

For the other identity, expand $I(R;E'S|Q'F)_1$ again in two possible ways:
\begin{align}
I(R;S|Q'F)_1
+
I(R;E'|Q'SF)_1
&=
I(R;E'|Q'F)_1
+
I(R;S|Q'E'F)_1.
\label{eq:environment-chain-comparison}
\end{align}
Again, unitary invariance and the initial product structure give
\[
I(R;S|Q'E'F)_1
=
I(R;S|QEF)_0
=
I(R;S|Q)_0,
\]
which proves Eq.~\eqref{eq:environment-to-two-sided-identity}.
Example~\ref{ex:environment-conditioning-instability} below shows that the
residual term can make $I(R;E'|Q'SF)_1$ strictly positive even when
$I(R;E'|Q'F)_1=0$.
\ensuremath{\blacksquare}%

Equation~\eqref{eq:system-decoupling-lifts} is stronger than the existential
implication System $\Rightarrow$ Global. It holds in the same fixed model, with
its possibly nontrivial $F$, and it implies that the given joint $SF$ extension
explains every later revival from $t_1$.

Preparation screening gives a simple initial condition for this robust property.

\begin{theorem}{Corollary}(Preparation screening implies robust
system-side Markov screening)
\label{cor:system-side-screening-criterion}
Let $\widetilde{\rho}_{SRQ}$ be an inert system-side extension satisfying
\begin{equation}
I(R;Q|S)_0=0.
\label{eq:general-system-screening}
\end{equation}
Then every compatible environment-side extension $\widetilde{\omega}_{EF}$ gives
a fixed extended LIM satisfying
\begin{equation}
I(R;E'|Q'SF)_1=0.
\label{eq:screening-two-sided-zero}
\end{equation}
Thus the given joint $SF$ extension explains every later revival within that
model. In particular, taking $F$ trivial gives a system-side explanation.
\nobreak\hfill\ensuremath{\square}\end{theorem}

\noindent{\it Proof.}
First take $F$ to be trivial. Initial factorization and unitary invariance give
\[
I(R;Q'E'|S)_1
=
I(R;QE|S)_0
=
I(R;Q|S)_0
=
0.
\]
The chain rule then implies $I(R;E'|Q'S)_1=0$. For arbitrary compatible $F$,
tracing out $F$ recovers the same $SRQ'E'$ marginal, so the same equality still
holds. Equation~\eqref{eq:system-decoupling-lifts} now gives
Eq.~\eqref{eq:screening-two-sided-zero}. The QCMI criterion in
Remark~\ref{rem:QCMI-explainability} completes the proof.
\ensuremath{\blacksquare}%

The reverse stability already fails for classical states and a swap interaction.
The example is intentionally simple: $I(R;Q)_0=0$, but conditioning on $S$
reveals an $R$--$Q$ correlation that is otherwise hidden.

\begin{definition}{Example}(Failure of environment-side Markov screening)
\label{ex:environment-conditioning-instability}
Let $R$, $S$, $Q$, $E$, and $F$ be classical bits, represented by diagonal
states. Set
\begin{align}
\widetilde{\rho}_{SRQ}
=
\frac{1}{4}\bigl(&
\ket{000}\!\bra{000}
+
\ket{001}\!\bra{001}
+
\ket{010}\!\bra{010}
+
\ket{111}\!\bra{111}
\bigr)_{SRQ}\;,
\label{eq:classical-asymmetry-state}
\end{align}
where the order is $S,R,Q$, and let
\begin{equation}
\widetilde{\omega}_{EF}
=
\frac{I_E}{2}\otimes\frac{I_F}{2}.
\label{eq:classical-asymmetry-environment}
\end{equation}
Let the first interaction swap $Q$ and $E$, so $Q'=E$ and $E'=Q$. The reduced
initial state is
\[
\rho_{RQ}
=
\frac{I_R}{2}\otimes\frac{I_Q}{2},
\]
and therefore $I(R;Q)_0=0$. Since $F$ is independent and inert,
\begin{align}
I(R;E'|Q'F)_1
&=
I(R;E'|Q')_1
=
I(R;Q|E)_0
=
I(R;Q)_0
=
0.
\label{eq:example-environment-side-zero}
\end{align}
Thus exact environment-side screening holds.

However, conditioning also on $S$ reveals the hidden correlation:
\begin{align}
I(R;E'|Q'SF)_1
&=
I(R;Q|ESF)_0
=
I(R;Q|S)_0
=
\frac{3}{4}\log 3-\log 2
>
0.
\label{eq:example-two-sided-positive}
\end{align}
Hence adding the correlated system-side register destroys exact environment-side
screening. The effect is classical and requires only a swap.
\nobreak\hfill\ensuremath{\square}\end{definition}

Although adding $S$ may destroy exact environment-side screening, the joint $SF$
extension may still explain a particular revival. The following proposition
gives the exact criterion.

\begin{theorem}{Proposition}(Explanation criterion after environment-side Markov
screening)
\label{prop:environment-side-two-sided-criterion}
Fix a compatible extended LIM with inert $S$ and $F$, and assume
\begin{equation}
I(R;E'|Q'F)_1=0.
\label{eq:environment-side-decoupling}
\end{equation}
Then the joint $SF$ extension explains the selected revival within this fixed
LIM if and only if
\begin{equation}
I(R;S|Q''F)_2
-
I(R;S|Q'F)_1
\leqslant
I(R;E''|Q''F)_2.
\label{eq:environment-side-residual-criterion}
\end{equation}
\nobreak\hfill\ensuremath{\square}\end{theorem}

\noindent{\it Proof.}
By Remark~\ref{rem:QCMI-explainability}, the joint extension explains the
revival exactly when
\begin{equation}
I(R;E''|Q''SF)_2
\geqslant
I(R;E'|Q'SF)_1.
\label{eq:two-sided-condition-fixed-extension}
\end{equation}
Expanding both sides by the chain rule gives
\begin{align}
I(R;E''|Q''F)_2
+
I(R;S|Q''E''F)_2
-
I(R;S|Q''F)_2
\nonumber\\
\geqslant{}
I(R;E'|Q'F)_1
+
I(R;S|Q'E'F)_1
-
I(R;S|Q'F)_1.
\label{eq:expanded-two-sided-condition}
\end{align}
The second interaction is unitary on $Q'E'$ and leaves $R,S,F$ inert, so
\[
I(R;S|Q''E''F)_2
=
I(R;S|Q'E'F)_1.
\]
These terms cancel. Using Eq.~\eqref{eq:environment-side-decoupling} leaves
exactly Eq.~\eqref{eq:environment-side-residual-criterion}.
\ensuremath{\blacksquare}%

The left-hand side of Eq.~\eqref{eq:environment-side-residual-criterion} is the
increase in the $R$--$S$ correlation revealed by conditioning on the accessible
system and $F$. The right-hand side is the final information about $R$ that
remains in the active environment beyond $Q''F$. Thus the fixed joint extension
explains the revival exactly when the former does not exceed the latter. This
condition is automatic if $I(R;S|Q''F)_2\leqslant I(R;S|Q'F)_1$.

There is no conflict with the existential inclusion structure. Under the
hypothesis $I(R;E'|Q'F)_1=0$, discarding $S$ gives an environment-side
explanation with the same $F$, and adjoining a trivial $S$ gives a global
explanation. The proposition asks the stricter question of whether the already
present joint $SF$ extension explains the selected revival.

\begin{definition}{Remark}(Mixing, memory, and convexity)
\label{rem:mixing}
A long-standing puzzle is that mixing
processes that are individually revival-free
may produce an information
revival~\cite{buscemi-2025-non-causal-backflow}. In our framework,
Definition~\ref{def:eLIM} requires the
initial state to have the product form
$\widetilde{\rho}_{SRQ}\otimes\widetilde{\omega}_{EF}$.
A mixture of such states need not
remain a product state when both
$\widetilde{\rho}_{SRQ}$ and $\widetilde{\omega}_{EF}$ vary. Two cases
nevertheless allow convexity to be
established.

First, fix a probe $\rho_{RQ}$ and consider
a finite mixture of admissible triples on
the same reference and system spaces. Let $x$ label
the selected process and $p_x$ its
probability. Suppose each process is revival-free
or environment-side explainable. For each process,
choose a compatible LIM and, when needed, an
explanatory register $F$. A classical
register in the environment stores $x$ and
controls both interactions of the selected
LIM, while an inert copy $C$ of $x$ is
retained alongside $F$. Because the
selection is independent of the probe, the
initial state retains the required product form.

Second, fix a LIM and mix probes $\rho^{(x)}_{RQ}$.
Suppose each resulting process is revival-free or
has a system-side explanation within this
same LIM. The label $x$ can now be stored in
an inert classical register $C$ alongside
the corresponding explanatory register $S$.
The environment state and both interactions
are common to all processes, so the initial
product form is again preserved.

In these two constructions, write $X=F$ and $X=S$,
respectively, and take $X$ trivial for 
processes without a revival. The finitely 
many environments and inert registers can be 
embedded into common spaces when necessary. 
Since the $RC$ marginal is unchanged between $t_1$ and $t_2$,
the contribution $I(R;C)$ cancels when 
comparing the two times. The chain rule and 
the classicality of $C$ therefore give
\begin{align}
&I(R;Q'XC)_1-I(R;Q''XC)_2 \notag\\
&\quad=\sum_x p_x\left[I(R;Q'X)^{(x)}_1-I(R;Q''X)^{(x)}_2\right]\geqslant0.
\end{align}
Consequently, the mixture is either revival
free or environment-side explainable in the 
first case, and either revival-free or 
system-side explainable within the fixed LIM 
in the second.

For a fixed probe, the record $x$ is a
memory of which dynamics was selected,
independent of the initial preparation. Its
inert copy $C$ never participates in the
interactions. Retaining this copy explains
revivals that arise only after the record
has been discarded, as in Ref.~\cite{buscemi-2025-non-causal-backflow}.

For general mixtures, the same label may
need to select both the extended preparation
and the environment state or interactions.
Retaining that label on both sides generally
introduces initial correlations across the
system--environment cut, contrary to
Definition~\ref{def:eLIM}. This failure of
the construction does not prove non
convexity of the class without genuine
backflow, since another compatible extended
LIM might explain the mixture. Whether
admissible triples that are revival-free or
globally explainable are closed under 
arbitrary convex mixtures that remain 
admissible is therefore left open.
\nobreak\hfill\ensuremath{\square}\end{definition}

\section{Numerical illustrations of explanatory mechanisms}
\label{sec:numerical-illustrations}

We conclude with four numerical illustrations. The first three use independently
selected extended LIMs to show three ways in which inert registers can
compensate a reduced revival. The fourth varies the initial probe and
illustrates the role of probe entanglement. These are existence examples, not
estimates of the frequency of the different mechanisms. Full-precision 
states, Hamiltonians,
seeds, arrays, configurations, and validation records are available in the
archive cited in the~\hyperref[sec:archived-numerical-release]{data 
and code availability} statement.

\subsection{Fixed-extension compensation patterns}
\label{subsec:representative-scenarios}

In the first three examples, all systems are qubits, the initial state has the
compatible form $\widetilde{\rho}_{SRQ}\otimes\widetilde{\omega}_{EF}$, and only
$Q$ and $E$ interact. Within each example they evolve under a single
time-independent Hamiltonian,
\[
W(t)=e^{-itH}.
\]
Once this continuous trajectory has been fixed, the intermediate and final times
may be chosen arbitrarily, subject only to $t_0<t_1<t_2$; in particular, no
relation such as $t_2=2t_1$ is required. For the three displayed examples we
choose, only for convenience,
\[
t_0=0,
\qquad
t_1=2,
\qquad
t_2=4.
\]
The two intervals consequently have the same duration in these particular plots,
so that the same unitary $e^{-2iH}$ acts during both of them. The three examples
use independently selected initial states and Hamiltonians; they are not three
trajectories of one physical model.
\begin{figure}[t]
\centering
\includegraphics[width=\textwidth]{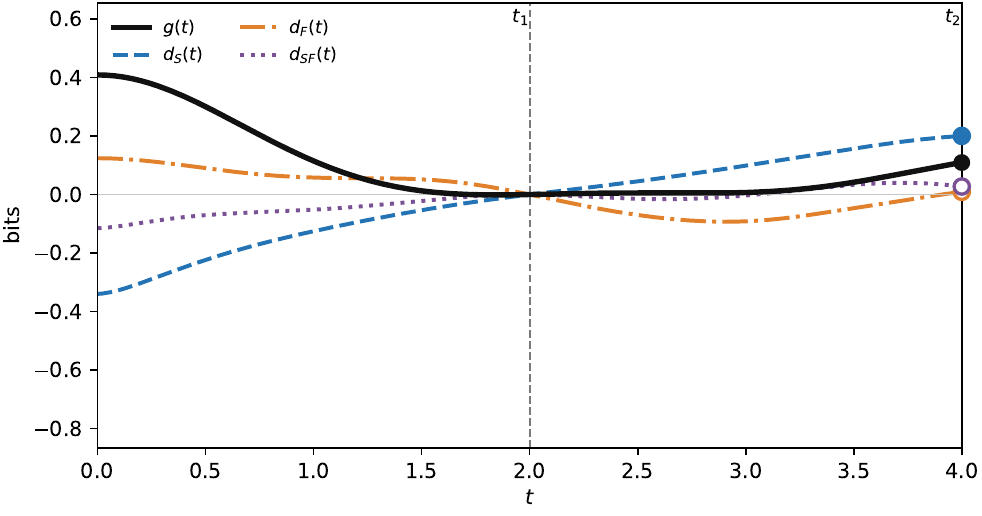}
\caption{Case A: system-side compensation in a fixed extended LIM. At $t_2$,
$d_S$ exceeds the reduced gain, whereas $d_F$ and $d_{SF}$ do not. Thus the
displayed $S$ explains the revival, while the displayed $F$ and nontrivial $SF$
do not.}
\label{fig:system-world-compensation}
\end{figure}
\FloatBarrier
For a retained inert register $X\in\{S,F,SF\}$, define the reduced gain and its
compensation by
\begin{align}
g(t)&:=I(R;Q_t)-I(R;Q')_1,
&
d_X(t)&:=I(R;X|Q')_1-I(R;X|Q_t).
\label{eq:numerical-gain-compensation}
\end{align}
The chain rule gives
\begin{equation}
I(R;Q_tX)-I(R;Q'X)_1=g(t)-d_X(t).
\label{eq:numerical-gain-compensation-identity}
\end{equation}
All the plotted curves therefore vanish at $t_1$. At $t_2$, the black curve
$g(t)$ gives the reduced revival, while a colored curve at or above it means
that the corresponding displayed register explains that revival. Only the
interval $t\geqslant t_1$ is relevant to this comparison. On the 
compensation curves,
filled endpoint markers indicate this condition, while open markers indicate a
positive residual gain. The continuous trajectories also show that the
explanatory power of a fixed inert extension may depend on the chosen final
time. Keeping $t_1$ and the extended LIM fixed, any $t>t_1$ for which $g(t)>0$
may be regarded as the endpoint of a revival, and the displayed register $X$
explains that revival exactly when
\[
d_X(t)\geqslant g(t).
\]
As the two curves may cross during the evolution, the same inert extension may
explain a revival ending at one time but fail to explain a revival ending at a
later time, or conversely. Explainability within a fixed extension is therefore
a property of the selected pair of times, rather than a permanent property of
the extension alone. The three fixed-extension patterns are displayed in
Figs.~\ref{fig:system-world-compensation},
\ref{fig:environment-world-compensation}, and \ref{fig:two-sided-compensation}.
They show, respectively, compensation by the displayed system-side register, by
the displayed environment-side register, and only by the displayed joint
register.
\begin{figure}[t]
\centering
\includegraphics[width=\textwidth]{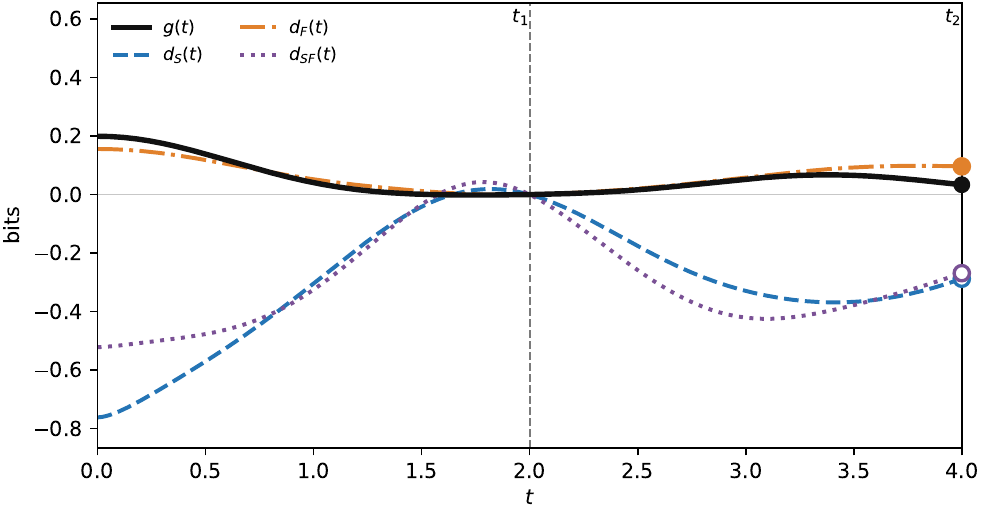}
\caption{Case B: environment-side compensation in a fixed extended LIM. At
$t_2$, only $d_F$ exceeds the reduced gain. Thus only the displayed $F$ explains
the revival.}
\label{fig:environment-world-compensation}
\end{figure}
\FloatBarrier

These conclusions concern only the displayed registers in the fixed extended
LIM. A curve below $g(t_2)$ does not rule out a different compatible extension.
In particular, the failure of the displayed nontrivial $SF$ register in Cases A
and B does not conflict with global explainability: the successful one-sided
explanation is itself a global explanation when the register on the opposite
side is taken to be trivial. The figures concern compensation of the selected
endpoint revival, not exact Markov screening at $t_1$.
\subsection{Isotropic probes and the reduced-gain map}
\label{subsec:entanglement-landscape}

We next fix the interaction and vary the initial probe. Consider the isotropic
family
\begin{equation}
\rho_{RQ}^{(\eta)}
=
\eta\,\ket{\Phi^+}\!\bra{\Phi^+}_{RQ}
+(1-\eta)\frac{I_{RQ}}{4},
\qquad
0\leqslant\eta\leqslant1.
\label{eq:isotropic-entanglement-family}
\end{equation}
Its preparable set is noncommuting for every $\eta>0$, whereas the probe is
separable exactly when $\eta\leqslant1/3$. We take $\omega_E=I_E/2$ and use, for
every $\eta$, the exchange Hamiltonian
\begin{equation}
H_{\mathrm{ex}}
=
\frac{\pi}{8}
\bigl(
\sigma_x\otimes\sigma_x
+\sigma_y\otimes\sigma_y
+\sigma_z\otimes\sigma_z
\bigr),
\qquad
W_{\mathrm{ex}}(t)=e^{-itH_{\mathrm{ex}}}.
\label{eq:landscape-hamiltonian}
\end{equation}
This interaction performs a full SWAP at $t_1=2$ and returns to the identity at
$t=4$, up to global phases. The quantity shown in the left panel of
Fig.~\ref{fig:entanglement-revival-landscape} is
\begin{equation}
r(\eta,t)
:=
\max\bigl\{0,I(R;Q_t)_\eta-I(R;Q_{t_1})_\eta\bigr\}.
\label{eq:reduced-revival-landscape}
\end{equation}
Only $t>t_1$ is
interpreted as a possible final time; the region $t<t_1$ is shown to display the
preceding loss of correlation.
\begin{figure}[t]
\centering
\includegraphics[width=\textwidth]{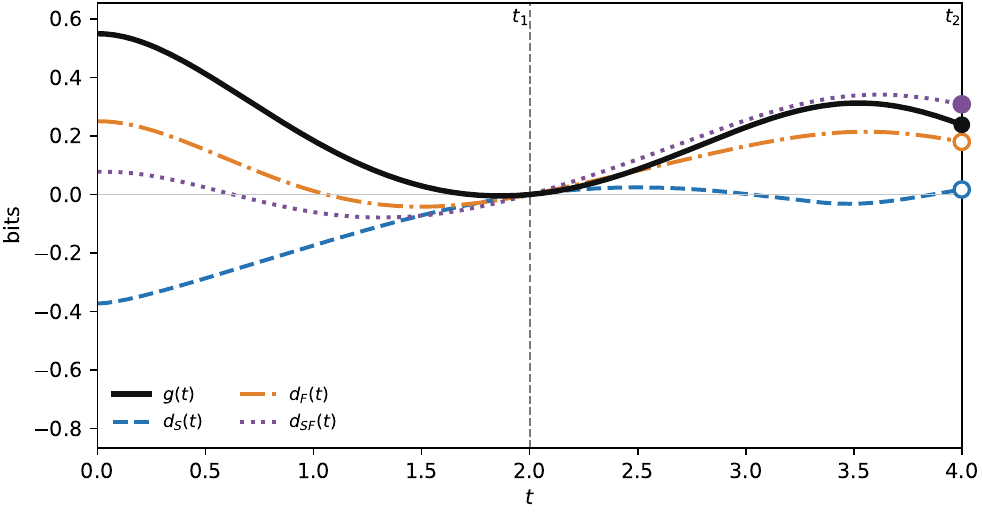}
\caption{Case C: joint compensation in a fixed extended LIM. At $t_2$, $d_{SF}$
exceeds the reduced gain, whereas $d_S$ and $d_F$ do not. Thus the displayed
joint register explains the revival, while neither displayed one-sided register
does.}
\label{fig:two-sided-compensation}
\end{figure}
\FloatBarrier
The right panel sets $r(\eta,t)$ to zero whenever at least one tested inert
extension has an extended gain no larger than the numerical zero tolerance
$10^{-3}$ bits. For $\eta\leqslant1/3$, the calculation includes the 
exact five-label
system-side screening extension obtained from a tetrahedral separable
decomposition of the isotropic probe. This extension satisfies $I(R;Q|S)_0=0$
and therefore explains every revival, by
Corollary~\ref{cor:system-side-screening-criterion}. For $\eta>1/3$, the test
uses a finite family of compatible extensions; it is not an
optimization over all compatible extensions. The complete construction,
finite-family specification, and validation records are included in 
the \hyperref[sec:archived-numerical-release]{archived numerical 
release}. Figure~\ref{fig:entanglement-revival-landscape} compares the
reduced-gain map with the same map after setting to zero the points explained by
the tested inert extensions.

\begin{figure}[t]
\centering
\includegraphics[width=\textwidth]{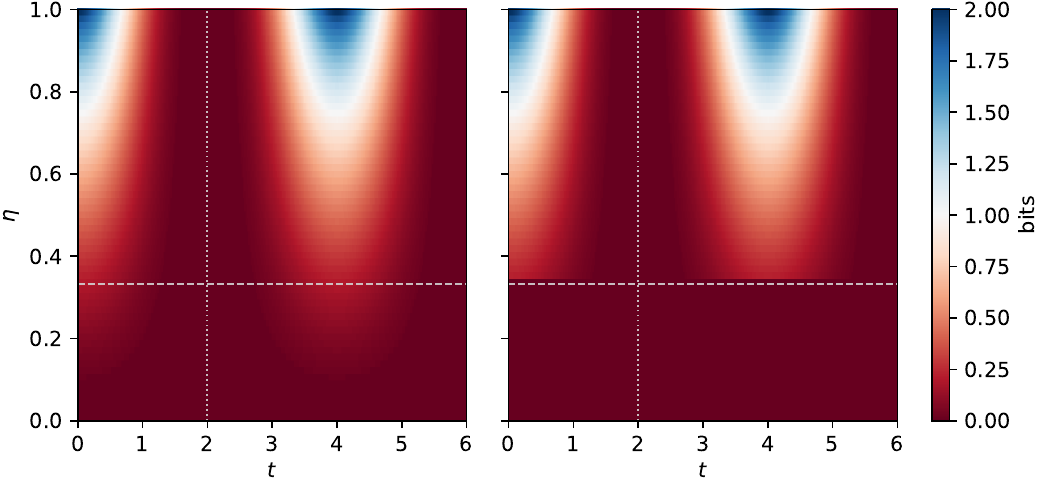}
\caption{Two-parameter reduced-gain map for the isotropic probes in
Eq.~\eqref{eq:isotropic-entanglement-family} under the exchange Hamiltonian in
Eq.~\eqref{eq:landscape-hamiltonian}. Left: the positive reduced gain relative
to $t_1=2$. Right: the same gain after points explained by the tested extensions
are set to zero. The vertical dotted line marks $t_1=2$, and the horizontal
dashed line marks the entanglement threshold $\eta=1/3$. The region at or below
that threshold vanishes in the right panel because the exact screening extension
is included. Positive values above the threshold in that panel mean only that
the tested finite family did not provide an explanation; they do not establish
genuine backflow.}
\label{fig:entanglement-revival-landscape}
\end{figure}

Thus the interval $0<\eta\leqslant1/3$ gives a direct illustration of 
the distinction
between the two necessary conditions found in Section~\ref{sec:results}: the
preparable set is already noncommuting, but the probe is separable and every
revival is extrinsic. Noncommutativity is therefore necessary, but not
sufficient, for genuine backflow.
\section{Conclusion}

We extended the distinction between extrinsic revivals and genuine information
backflow from quantum channels to open-system processes with restricted initial
preparations. The preparable states are encoded through a reference--system
probe, which determines both the available inputs and the information monitored
during the evolution.

Our central structural result is exact: a probe is quantum--classical on the
system if and only if its preparable set is commuting. This removes any
ambiguity associated with the choice of probe in the classical case. Since every
quantum--classical probe is separable, and every separable probe admits an inert
classical extension that explains all its revivals, genuine backflow under the
global convention requires both an entangled probe and a noncommuting preparable
set. Neither condition is sufficient by itself.

We also established the V-shaped hierarchy among system-side, environment-side,
and global explanations, and identified an asymmetry between the stability of
system-side and environment-side screening within a fixed extended model.
Together, these results separate the observation of an information revival from
its dynamical interpretation. In particular, classical information can revive,
but such a revival cannot by itself certify that information has returned from
the environment.

\section*{Data and code availability}
\label{sec:archived-numerical-release}
The source code and numerical data supporting this work are openly available in
Zenodo, version~1.0.1, at
\href{https://doi.org/10.5281/zenodo.22702111}{%
\nolinkurl{10.5281/zenodo.22702111}}, under the MIT
License.

\section*{Declarations}

\noindent\textit{Competing interests.}
The authors declare no competing interests.

\medskip
\noindent\textit{Use of AI-assisted tools.}
OpenAI ChatGPT/Codex was used for language editing, LaTeX assistance, and review
of the numerical-code presentation. The authors independently verified the
mathematical arguments, numerical results, citations, and final text and take
full responsibility for the content.

\section*{Acknowledgments}

The authors acknowledge support from the MEXT Quantum Leap Flagship Program
(MEXT QLEAP), Grant No.~JPMXS0120319794, and from JSPS KAKENHI, Grant
Nos.~23K03230 and 26K00621.

{\small
\setlength{\itemsep}{0pt}
\setlength{\parsep}{0pt}
}

\end{document}